\documentclass[12pt]{amsart}
\usepackage{amssymb,amsmath,amsfonts,   amsthm,nomencl,mathrsfs } 
\usepackage[arrow, matrix, curve]{xy}
\usepackage{a4wide}
\usepackage{color}
\newcommand{\IC}{\mathbb{C}}
\newcommand{\IR}{\mathbb{R}}
\newcommand{\IEE}{\mathscr{E}}

\newcommand{\question}[1]{\leavevmode{\marginpar{\tiny%
$\hbox to 0mm{\hspace*{-0.5mm}$\leftarrow$\hss}%
\vcenter{\vrule depth 0.1mm height 0.1mm width \the\marginparwidth}%
\hbox to 0mm{\hss$\rightarrow$\hspace*{-0.5mm}}$\\\relax\raggedright #1}}}

\newcommand{\ILL}{\mathscr{L}}

\newcommand{\IFF}{\mathcal{F}}

\newcommand{\dom}{\mathrm{Dom}}

\newcommand{\m}{\mathfrak{m}}

\newcommand{\PP}{\mathfrak{P}}

\newcommand{\IN}{\mathbb{N}}

\newcommand{\f}{\frac}
\newcommand{\nn}{\nonumber}

\theoremstyle{plain}            
\newtheorem{theorem}{theorem}[section]
\newtheorem{Lemma}[theorem]{Lemma}
\newtheorem{Corollary}[theorem]{Corollary}
\newtheorem{Theorem}[theorem]{Theorem}
\newtheorem{Proposition}[theorem]{Proposition}

\theoremstyle{definition}       
\newtheorem{Definition}[theorem]{Definition}
\newtheorem{A^{x,y}ssumption}[theorem]{A^{x,y}ssumption}
\newtheorem{Remark}[theorem]{Remark}

\allowdisplaybreaks[1]

\begin{document}

\setcounter{page}{1}

\begin{titlepage}

\author{Batu G\"uneysu}
\address{Institut für Mathematik, Humboldt-Universit\"at zu Berlin, Rudower Chaussee 25, 12489 Berlin}
\email{gueneysu@math.hu-berlin.de}

\title[Geometry of molecular Schr\"odinger semigroups]{$\mathrm{RCD}^*(K,N)$ spaces and the geometry of multi-particle Schr\"odinger semigroups}

\end{titlepage}

\begin{abstract} Let $(X,\mathfrak{d},\m)$ be an $\mathrm{RCD}^*(K,N)$ space for some $K\in\IR$, $N\in [1,\infty)$, let $H$ be the self-adjoint Laplacian induced by the underlying Cheeger form. Given $\alpha\in [0,1]$ we introduce the $\alpha$-Kato class of potentials on $(X,\mathfrak{d},\m)$, and given a potential $V:X\to \IR$ in this class, we denote with $H_V$ the natural self-adjoint realization of the Schr\"odinger operator $H+V$ in $L^2(X,\m)$. We use Brownian coupling methods and perturbation theory to prove that for all $t>0$ there exists an explicitly given constant $A(V,K,\alpha,t)<\infty$, such that for all $\Psi\in L^{\infty}(X,\m)$, $x,y\in X$ one has
\begin{align*}
\big|e^{-tH_V}\Psi(x)-e^{-tH_V}\Psi(y)\big|\leq   A(V,K,\alpha,t) \|\Psi\|_{L^{\infty}}\mathfrak{d}(x,y)^{\alpha}.
\end{align*}
In particular, all $L^{\infty}$-eigenfunctions of $H_V$ are globally $\alpha$-H\"older continuous. This result applies to multi-particle Schr\"odinger semigroups and, by the explicitness of the H\"older constants, sheds some light into the geometry of such operators.
\end{abstract}  

\maketitle 

\vspace{3mm}

\emph{\Large{Dedicated to the memory of Kazumasa Kuwada.}}

\vspace{3mm}
\section{Introduction}

Among several other fundamental results, T. Kato proved the following result in 1957 in his seminal paper on mathematical quantum mechanics \cite{Kato}: consider the multi-particle Schr\"odinger operator $H_V=-\Delta+V$ in $L^2(\IR^{3m})$ with a potential $V:\IR^{3m}\to \IR$ of the form
$$
V(x)=\sum_{1\leq j\leq m } V_{j}(\mathbf{x}_j)+ \sum_{1\leq j<k\leq m } V_{jk}(\mathbf{x}_j-\mathbf{x}_k)\quad \text{with $V_j,V_{jk}\in L^{q}(\IR^{3})+L^{\infty}(\IR^{3})$}
$$
for some $q\geq 2$. Then all eigenfunctions (or more generally, elements of $\cap_n\dom(H_V^n)$, where $H_V^n$ denotes the $n$-th power of $H_V$) of $H_V$ are globally $\alpha$-H\"older continuous for all $0<\alpha< 2-3/q$. In particular, applying this result to 
$$
V_j(\mathbf{x}):=-\sum^l_{i=1}Z_i/|\mathbf{x}-\mathbf{R}_i|,\quad V_{jk}(\mathbf{x}):=1/|\mathbf{x}|
$$
one finds that the eigenfunctions of any molecular Schr\"odinger operator having $m$ electrons, $l$ nuclei and having its $i$-th nucleus sitting in $\mathbf{R}_i$ with $\sim Z_i$ protons is globally $\alpha$-H\"older continuous for all $0<\alpha< 1$. The fact that the above class of potentials depends on the \emph{linear} surjective maps 
$$
\pi_{ij}:\IR^{3m}\longrightarrow \IR^{3},\quad \pi_{ij}(\mathbf{x}_1,\dots,\mathbf{x}_m):= \mathbf{x}_i-\mathbf{x}_j, 
$$
and the fact that Kato's proof heavily relies on the Fourier transform, raise the following natural question: \vspace{2mm}

\emph{What is the geometry behind Kato's global regularity result?}\vspace{2mm}

Rather than studying regularity properties of eigenfunctions, we take the approach of studying $L^{\infty}\to C^{0,\alpha}$ mapping properties of the underlying \emph{Schr\"odinger semigroup} $e^{-tH_V}$, $t>0$, following B. Simon's seminal paper \cite{Simon}. This procedure relies on the spectral theorem, showing that $H_V\Psi =\lambda\Psi$ implies $\Psi=e^{t\lambda} e^{-tH_V}\Psi$ for all $t\geq 0$.\vspace{2mm}

The global mapping property $L^{\infty}\to C^{0,\alpha}$ of a Schr\"odinger semigroup $e^{-tH_V}$ is well known to be delicate even for $V=0$ in Riemannian geometry: indeed, the heat semigroup on an arbitrary noncompact Riemannian manifold $M$ needs not be globally Lipschitz smoothing, while it is, if $M$ is complete with a Ricci curvature bounded from below.  \vspace{1mm}

Being motivated by this fact, we pick $\mathrm{RCD}^*(K,N)$ spaces as our state spaces in this paper, where $K\in\IR$, $N\in [1,\infty)$. These are metric measure spaces having a linear heat flow, and a Ricci curvature bounded from below by $K$ and dimension $\leq N$ in the sense of $\mathrm{CD}^*(K,N)$ spaces, that have been introduced by Bacher/Sturm in \cite{bacher}. It should be noted here that, in general, the $\mathrm{CD}^*(K,N)$ implies the original $\mathrm{CD}(K,N)$ condition by Sturm \cite{sturm} (see also Lott/Villani \cite{lott}), while recently Cavalletti/Milman \cite{cavaletti2} have shown that in case the underlying measure is finite, the $\mathrm{RCD}^*(K,N)$ condition is equivalent to the conjunction of the $\mathrm{CD}(K,N)$ condition and the linearity of the heat flow (this is the $\mathrm{RCD}(K,N)$ condition). Moreover, the $\mathrm{RCD}^*(K,N)$ condition implies \cite{erbar} the $\mathrm{RCD}(K,\infty)$ condition, which has been originally introduced by Ambrosio/Gigli/Savar\'{e} \cite{ags2}. The essential feature of these spaces for our purposes is that the underlying (linear) heat semigroup is globally Lipschitz smoothing \cite{ags2,ags}. We refer the reader to Remark \ref{sswqq} below and the references therein for examples of such spaces.\vspace{1mm}

Given an $\mathrm{RCD}^*(K,N)$ space $X\equiv (X,\mathfrak{d},\m)$ and $0\leq\alpha\leq 1$ we introduce a new class $\mathcal{K}^{\alpha}(X)$ of potentials $V:X\to\IR$ which we call the \emph{$\alpha$-Kato class} of $(X,\mathfrak{d},\m)$ (cf. Definition \ref{kat}), and which refines the usual Kato class $\mathcal{K}(X)$ in the sense that $\mathcal{K}(X)=\mathcal{K}^{0}(X)$. For any $V\in \mathcal{K}^{\alpha}(X)$, $\alpha\in (0,1]$ it turns out that one can naturally define a Schr\"odinger semigroup $e^{-tH_V}$ in $L^2(X)$. Our main result (cf. Theorem \ref{haupt1} below) states that this Schr\"odinger semigroup has the smoothing property 
$$
e^{-tH_V}:L^{\infty}(X)\to C^{0,\alpha}(X)\quad\text{ for all $t>0$},
$$
with a H\"older constant depending explicitly on $K$, $\alpha$, $t$, $V$, but not on $N$. The fact that the Hölder constant does not depend on $N$ indicates that our main result should actually hold true on $\mathrm{RCD}(K,\infty)$ spaces. The finite dimensionality enters our proof only through the pointwise existence of the Wiener measures (cf. Remark \ref{asllp} below).\vspace{1mm}

Theorem \ref{haupt1} is fundamentally new even for the $\mathrm{RCD}^*(0,N)$ space induced by the Euclidean $\IR^N$ with its Lebesgue measure, in the sense that it can deal with more general potentials than in Kato's paper and that it provides explicit constants. To the best of our knowledge, Theorem \ref{haupt1} can only be compared with Theorem B.3.5 in \cite{Simon}, where B. Simon has proved a \emph{local} $\alpha$-H\"older smoothing result for Schr\"odinger semigroups in $\IR^N$ under an $\alpha$-dependent Kato type assumption on the potential.\vspace{1mm}

As an application of our main result, we prove a H\"older smoothing result for multi-particle Schr\"odinger semigroups with $L^q$-potentials on Riemannian manifolds having a Ricci curvature bounded from below by a constant, where the above maps $\pi_{ij}$ are replaced by very general surjective Riemannian submersions. This yields a natural generalization of Kato's result to the Riemannian setting (cf. Corollary \ref{sub}) and ultimately explains the geometry behind Kato's Euclidean result.\vspace{1mm}

Finally, we apply our results to produce explicit $L^r(\IR^{3m})\to C^{0,\alpha}(\IR^{3m})$ bounds for the Schr\"odinger semigroup $e^{-tH_{V_{R,Z}}}$ corresponding to a molecule (here, as bove, $R$ is the location vector of the nuclei and $Z$ the corresponding charge vector), where $r\in [1,\infty]$, $\alpha\in (0,1)$. Our methods produce an explicit H\"older constant that has an $\alpha$-dependence of the form 
$$
C_{m,Z}t^{-\frac{3m}{2r}}e^{C_{R,Z}t}\Big(2^{1-\alpha}t^{-\alpha/2}+   \frac{(t/4)^{-\alpha/2+1/2}}{1/2-\alpha/2}e^{C_{R,Z}t}\Big),
$$
which shows that one cannot take $\alpha \nearrow 1$ in order to obtain $L^r(\IR^{3m})\to C^{0,1}(\IR^{3m})$ (= Lipschitz) estimates. This fact is consistent with Kato's result and the fact that in this case
$$
V_{R,Z}\in\Big( \bigcap_{\alpha\in (0,1)}\mathcal{K}^{\alpha}(\IR^{3m}) \Big)\setminus \mathcal{K}^{1}(\IR^{3m}).
$$
On the other hand, Kato proves 'by hand' that molecular Schr\"odinger semigroups map $L^r(\IR^{3m})\to C^{0,1}(\IR^{3m})$, a fact which becomes rather mysterious from the point of view of our probabilistic methods and which raises the following open question:\vspace{2mm}

\emph{Is there a 'probabilistic' proof of the smoothing property $e^{-tH_{V_{R,Z}}}:L^r(\IR^{3m})\to C^{0,1}(\IR^{3m})$ with $V_{R,Z}$ the potential of a molecule as above? }

\vspace{2mm}

We close the introduction with some remarks concerning our proof of Theorem \ref{haupt1}, which comes in two steps: without making use of Fukushima's \cite{fuku} or Ma/R\"ockner's \cite{roeck} abstract theory for (quasi-) regular Dirichlet forms\footnote{which yields the existence of the Wiener measure for quasi-almost every initial point} we first show that \emph{for all} initial points $x\in X$ there is a natural Wiener measure or Brownian motion measure $\PP^x$ on the space of continuous paths $[0,\infty)\to  X$, which allows us to obtain maximal Brownian couplings for all pairs of initial points $x,y$ \footnote{We refer the reader to Section \ref{swwwa} below for the basics of couplings of diffusions.}. In particular, the coupling time of such a coupling can be estimated by $\mathfrak{d}(x,y)$ times the constant from the aforementioned $L^{\infty}(X)\to C^{0,1}(X)$ smoothing property of the heat semigroup on $X$, and this fact allows us to deduce that the heat semigroup on $X$ is $L^{\infty}(X)\to C^{0,\alpha}(X)$ smoothing for \emph{all} $0< \alpha\leq 1$. Being equipped with this self-improvement property, we use perturbation theory to show that 
$$
e^{-tH_V}:L^{\infty}(X)\to C^{0,\alpha}(X)\quad\text{ for $V$'s in $\mathcal{K}^{\alpha}(X)$}
$$
to complete the proof.\vspace{2mm}

\emph{Acknowledgements:} The author would like to thank L. Ambrosio, Z.-Q. Chen, K. Kuwae, G. Savar\'{e} and the anonymous referees for several corrections and helpful remarks.







\section{Couplings of diffusions}\label{swwwa}

In the sequel, given probability measures $\mu_1$ and $\mu_2$ on a measurable space $(\Omega,\IFF)$, we denote with $\mathscr{C}(\mu_1,\mu_2)$ the set of all couplings of $\mu_1$ and $\mu_2$, that is, the set of all probability measures $\mu$ on $(\Omega\times \Omega,\IFF\otimes \IFF)$ such that $(\pi_j)_*\mu= \mu_j$ for $j=1,2$, where $\pi_j:\Omega\times\Omega\to \Omega$ denotes the projection onto the $j$-th component. Furthermore, we denote with 
$$
\delta(\mu_1,\mu_2):= \sup_{B\in \IFF}|\mu_1(B)-\mu_2(B)|
$$
the total variation distance. \vspace{2mm}

Let $X\equiv (X,\mathfrak{d})$ be a complete separable locally compact metric space. The space $C([0,\infty),X)$ of continuous paths $\omega:[0,\infty)\to X$ is always equipped with its topology of locally uniform convergence. By a diffusion on $X$ we will simply understand a family of Borel probability measures $\PP=(\PP^z)_{z\in X}$ on $C([0,\infty),X)$ such that 
\begin{itemize}
	\item $\PP^z\{\omega:\omega(0)=z\}=1$ for all $z\in X$,
	\item the map $z\mapsto \PP^z(A)$ is Borel measurable, for all Borel sets $A\subset C([0,\infty),X)$, 
	\item $\PP$ has the strong Markov property with respect to the natural filtration of $C([0,\infty),X)$.
	\end{itemize}

In the above situation, given $x,y\in X$, a continuous $X\times X$-valued process 
$$
(\mathbb{X},\mathbb{Y}): (\Omega ,\IFF,P)\longrightarrow C([0,\infty), X\times X)
$$
is called a \emph{coupling from $x$ to $y$ of $\PP$}, if $(\mathbb{X},\mathbb{Y})_*P\in \mathscr{C}(\PP^x,\PP^y)$, and then the coupling is called \emph{successful}, if with the \emph{coupling time}
$$
\tau(\mathbb{X},\mathbb{Y}):=\inf\{s>0: \mathbb{X}_{s'}=\mathbb{Y}_{s'}\>\text{ for all $s'\geq s$}\}:\Omega\longrightarrow [0,\infty],
$$
one has $P\{\tau(\mathbb{X},\mathbb{Y})=\infty\}=0$, and \emph{maximal}, if one has
 \begin{align}\label{w1}
P   \{ \tau(\mathbb{X},\mathbb{Y})>t\}= \frac{1}{2}\delta\big((\mathbb{X}_t)_*P,(\mathbb{Y}_t)_*P\big)
 \text{ for all $t>0$.}
\end{align}

An abstract result by Sverchkov/Smirnov \cite{smirnov} for couplings of cadlag Markov processes states that for all $x,y\in X$ there exists a maximal coupling of $\PP$ from $x$ to $y$. The following result, which should be well known to the experts (cf. \cite{sturm0} for a special case of part a)), will allow us to switch from global Lipschitz-smoothing to global H\"older-smoothing results later on, a rather subtle business on noncompact spaces: 

\begin{Proposition}\label{kazu} a) Assume  
\begin{align}\label{abs}
(\mathbb{X},\mathbb{Y}): (\Omega, \IFF,P)\longrightarrow C([0,\infty), X\times X)
\end{align}
is a coupling of $\PP$ from $x\in X$ to $y\in X$. Then for all $t>0$ one has 
 $$
P   \{ \tau(\mathbb{X},\mathbb{Y})>t\}\geq \frac{1}{2}\delta\big((\mathbb{X}_t)_*P,(\mathbb{Y}_t)_*P\big).
$$
b) For every function $F:(0,\infty)\to (0,\infty)$ the following statements are equivalent:
\begin{itemize}
	\item[i)] For all $x,y\in X$ there exists a coupling  
	\begin{align}\label{abs2}
(\mathbb{X},\mathbb{Y}): (\Omega, \IFF,P)\longrightarrow C([0,\infty), X\times X)
\end{align}
of $\PP$ from $x$ to $y$ with 
	\begin{align}\label{aaaadre}
P   \{ \tau(\mathbb{X},\mathbb{Y})>t\}\leq \frac{1}{2}F(t)\mathfrak{d}(x,y)\quad\text{ for all $t>0$}.
\end{align}
\item[ii)] For all bounded $f:X\to \IR$, $t>0$, $x,y\in X$, one has  
$$
 \left|\int f(\omega(t)) \PP^x( d \omega)-\int f(\omega(t)) \PP^y( d \omega) \right|\leq F(t)\mathfrak{d}(x,y)\left\|f\right\|_{\infty}.
$$
\item[iii)] For all bounded $f:X\to \IR$, $t>0$, $x,y\in X$, $\alpha\in (0,1]$ one has  
$$
 \left|\int f(\omega(t)) \PP^x( d \omega)-\int f(\omega(t)) \PP^y( d \omega) \right|\leq   F(t)^{\alpha} 2^{1-\alpha} \mathfrak{d}(x,y)^{\alpha} \left\|f\right\|_{\infty}.
$$
\end{itemize}
\end{Proposition}

\begin{proof}a) Let $B\subset X$ be an arbitrary Borel set. Assuming w.l.o.g. that $P\{\mathbb{X}_t\in B\}\geq P\{\mathbb{Y}_t\in B\}$, we have 
\begin{align*}
P   \{ \tau(\mathbb{X},\mathbb{Y})>t\}\geq  P   \{ \mathbb{X}_t\neq \mathbb{Y}_t\}\geq P   \{ \mathbb{X}_t\in B, \mathbb{Y}_t\ne B \}\geq |P   \{ \mathbb{X}_t\in B \} -P   \{ \mathbb{Y}_t\in B \}|.
\end{align*}
b) ii) $\Rightarrow$ i): Given $x,y\in X$ we pick a maximal coupling (\ref{abs2}) of $\PP$ from $x$ to $y$. Then we have
\begin{align*}
P   \{ \tau(\mathbb{X},\mathbb{Y})>t\}= \frac{1}{2} \sup_{B\subset X, \text{$B$ Borel}} \left|\int 1_B(\omega(t)) \PP^x( d \omega)-\int 1_B(\omega(t)) \PP^y( d \omega) \right|\leq \frac{1}{2} F(t) \mathfrak{d}(x,y).
\end{align*}
i) $\Rightarrow$ ii): One can estimate as follows:
\begin{align}
\nn&\left|\int f(\omega(t)) \PP^x( d \omega)-\int f(\omega(t)) \PP^y( d \omega)\right|\\
\nn&=  \left|\int \big(f(\mathbb{X}_t)- f(\mathbb{Y}_t) \big) dP\right|\\
\nn&=\left|\int_{\{ \tau(\mathbb{X},\mathbb{Y})>t\}} \big(f(\mathbb{X}_t)- f(\mathbb{Y}_t) \big) dP\right|\\ 
\label{aapopoq}&\leq P   \{ \tau(\mathbb{X},\mathbb{Y})>t\} 2 \left\|f\right\|_{\infty}\\
\nn&\leq   F(t)\mathfrak{d}(x,y)\left\|f\right\|_{\infty}.
\end{align}
iii) $\Rightarrow$ ii): Trivial.\\
ii) $\Rightarrow$ iii): this follows from (\ref{aapopoq}), $P(\cdot)\leq 1$ and ii) $\Rightarrow$ i),
\begin{align*}
&\left|\int f(\omega(t)) \PP^x( d \omega)-\int f(\omega(t)) \PP^y( d \omega)\right|\\
&\leq P   \{ \tau(\mathbb{X},\mathbb{Y})>t\}^{\alpha} 2 \left\|f\right\|_{\infty}\\
&\leq   F(t)^{\alpha} 2^{1-\alpha} \mathfrak{d}(x,y)^{\alpha} \left\|f\right\|_{\infty},
\end{align*}
completing the proof.
\end{proof}

\section{Schr\"odinger semigroups on $\mathrm{RCD}^*(K,N)$ spaces}

Assume now $X\equiv (X,\mathfrak{d},\m)$ is a metric measure space, that is, $(X,\mathfrak{d})$ is a complete seperable geodesic metric space, $\m$ is a $\sigma$-finite Borel measure which is finite on all open balls $B(x,r)$. We assume also that $(X,\mathfrak{d})$ is a geodesic space and that $\m$ has a full support. In the sequel, as we have applications in quantum mechanics in mind, we understand the spaces $L^q(X)=L^q(X,\m)$ to be over $\IC$, and we consider sesquilinear maps to be antilinear in the first component. The $L^q(X)$-norms will simply be denoted with $\left\|\cdot\right\|_{L^q}$.\vspace{1mm}

The \emph{Cheeger form} 
$$
\mathrm{Ch}:L^2(X)\longrightarrow [0,\infty]
$$
on $X$ is defined to be the $L^2$-lower semicontinuous relaxation of the functional 
$$
\widetilde{\mathrm{Ch}}:L^2(X)\cap \mathrm{Lip}(X)\longrightarrow [0,\infty],\quad \widetilde{\mathrm{Ch}}(f):= \int_X   \Big( \limsup_{y\to x}\frac{|f(x)-f(y)|}{\mathfrak{d}(x,y)} \Big) ^2 \m(dx),
$$
where $\limsup_{y\to x}|f(x)-f(y)|\mathfrak{d}(x,y)^{-1}$ is set $0$ if $x$ is isolated. In other words, for all $f\in L^2(X)$ one sets
$$
\mathrm{Ch}(f):=\liminf_{n\to\infty}\left\{\widetilde{\mathrm{Ch}}(f_n): (f_n)_{n\in\IN}\subset L^2(X)\cap \mathrm{Lip}(X), \left\|f_n-f\right\|_{L^2}\to 0\>\text{as $n\to\infty$} \right\}
$$

One gets a functional $\IEE$ on $L^2(X)$ with domain of definition 
$$
\dom(\IEE):=\{f:\mathrm{Ch}(f)<\infty\}
$$ 
by setting 
$$
\IEE(f):=\mathrm{Ch}(f)\quad\text{for all $f\in\dom(\IEE)$}.
$$

\begin{Definition}
The metric measure space $X$ is called \emph{infinitesimally Hilbertian}, if $\IEE$ is a quadratic form. 
\end{Definition}

Let $\mathsf{P}(X)$ denote the set of all Borel probability measures on $X$ and let $\mathsf{P}_2(X)$ denote the elements of $\mathsf{P}(X)$ that have finite second moments. Given elements $\mu_0$ and $\mu_1$ in $\mathsf{P}_2(X)$, the $L^2$-\emph{Wasserstein distance} is given by
$$ 
W_2(\mu_0,\mu_1):=\inf\left\{\int_{X\times X}\mathfrak{d}(x,y)^2 q(dx,dy):\>\>q\in\mathscr{C}(\mu_0,\mu_1)\right\},
$$
and any minimizer of the above infimum is called and \emph{optimal coupling} of $\mu_0$ and $\mu_1$. Then $\mathsf{P}_2(X)$ together with the $L^2$-Wasserstein distance is again a complete seperable geodesic space (as $X$ is so). \vspace{1mm}

Given $K\in\IR$, $t\in [0,1]$, $N\in [1,\infty)$ define the function
$$
\sigma_{K,N}^{(t)}:[0,\infty)\longrightarrow \IR\cup \{\infty\}
$$
by
$$
\sigma_{K,N}^{(t)}(\theta):=
\begin{cases}
&\infty, \quad\text{if $K\theta^2\geq N\pi^2$}\\
&\frac{\sin(t\theta\sqrt{K/N})}{\sin(\theta\sqrt{K/N})},\quad \text{if $0<K\theta^2< N\pi^2$}\\
&t, \quad\text{if $K\theta^2=0$}\\
&\frac{\sinh(t\theta\sqrt{-K/N})}{\sinh(\theta\sqrt{-K/N})}, \quad\text{if $K\theta^2< 0$}.
\end{cases}
$$

\begin{Definition}
Given $K\in\IR$, $N\in [1,\infty)$, the metric measure space $X$ is called an \emph{$\mathrm{CD}^*(K,N)$ space}, if for all $\mu_0,\mu_1\in \mathsf{P}(X)$ with bounded support and $\mu_0,\mu_1\ll\m$, there exists an optimal coupling $q$ of them and a geodesic $(\mu_t)_{t\in [0,1]}$ in $\mathsf{P}_2(X)$ connecting them, such that
\begin{itemize}
\item $\mu_t\ll\m$ and $\mu_t$ has a bounded support for all $t\in [0,1]$,

\item for all $N'\in[N,\infty)$, $t\in [0,1]$ one has
\begin{align*}
&\int_X \rho_t(x)^{1-1/N'} \m(dx)\geq \\
&\int_{X\times X} \Big(\sigma_{K,N'}^{(1-t)}(\mathfrak{d}(x_0,x_1))\rho_0(x_0)^{-1/N'} +\sigma_{K,N'}^{(t)}(\mathfrak{d}(x_0,x_1))\rho_1(x_1)^{-1/N'}\Big)\>q(dx_0, dx_1), 
\end{align*}
where $\rho_t$ denotes the Radon-Nikodym density of $\mu_t$ with respect to $\m$ for all $t\in [0,1]$. 
\end{itemize}
\end{Definition}

\begin{Definition}
Given $K\in\IR$, $N\in [1,\infty)$, the metric measure space $X$ is called an \emph{$\mathrm{RCD}^*(K,N)$ space}, if it is an infinitesimally Hilbertian $\mathrm{CD}^*(K,N)$ space.
\end{Definition}

\begin{Remark}\label{sswqq}1. Several equivalent characterizations of the $\mathrm{RCD}^*(K,N)$ condition have been obtained by Erbar/Kuwada/Sturm \cite{erbar} and by Cavalletti/Milman \cite{cavaletti2}. In particular, it is shown in \cite{cavaletti2} that the $\mathrm{RCD}^*(K,N)$ condition, is equivalent to the $\mathrm{RCD}(K,N)$ condition, if the underlying measure is finite. Moreover, natural tensorization, stability, and local-to-global results of $\mathrm{RCD}^*(K,N)$ spaces have been established in \cite{bacher, erbar}. \vspace{1mm}

2. On a smooth geodesically complete Riemannian manifold $M$ equipped with its geodesic distance and its Riemannian volume measure one has $\dom(\IEE)=W^{1,2}(M)$ and 
$$
\IEE(f_1,f_2)=\int_M (\nabla f_1(x) ,\nabla f_2(x) ) \m(dx)   .
$$
Moreover, $M$ is $\mathrm{RCD}^*(K,N)$, if and only if $\mathrm{Ric}\geq K$ and $\dim(M)\leq N$ \cite{bacher}. In particular, in this case one has $H=-\Delta$, the (unique self-adjoint realization of the) Laplace-Beltrami operator. \vspace{1mm}

3. A new interesting class of possibly very singular $\mathrm{RCD}^*(K,N)$ spaces (having a finite measure) has been found by Bertrand/Ketterer/Mondello/Richard in \cite{ketterer}: there the authors show that every compact smoothly stratified space $X$ with an iterated edge metric on its regular part (which is a smooth open Riemannian manifold) canonically induces a metric measure space, which is $\mathrm{RCD}(K,N)$ if and only if $\dim(X)\leq N$ and $X$ has a singular Ricci curvature bounded from below by $K$ in the sense of \cite{ketterer}.
\end{Remark}

\emph{We assume from here on that $X$ is an $\mathrm{RCD}^*(K,N)$ space for some $K\in\IR$ and some $N\in [1,\infty)$.}
\vspace{3mm}

As pointed out by Kuwada/Kuwae \cite{kuwada}, then $X$ is locally compact (this relies on the validity of the Bishop-Gromov volume estimate \cite{bacher} on $\mathrm{RCD}^*(K,N)$ spaces; see also the proof of Lemma \ref{assa} below) and in fact $\IEE$ becomes a regular strongly local Dirichlet form in $L^2(X)$. This follows from combining the facts that any $\mathrm{RCD}^*(K,N)$ is an $\mathrm{RCD}(K,\infty)$ space \cite{erbar}, that the Cheeger energy on $\mathrm{RCD}(K,\infty)$ space is a quasi-regular Dirichlet form \cite{savare}, and that square-integrable Lipschitz functions are dense in the domain of definition of the Cheeger energy with respect to the energy norm \cite{ohje} (which ensures the regularity of $\IEE$).\vspace{1mm}

If we denote by $H$ the nonnegative self-adjoint operator in $L^2(X)$ corresponding to $\IEE$, \emph{the Laplacian on $X$}, the \emph{heat semigroup} 
$$
(e^{-t H})_{t\geq 0}\subset \ILL(L^2(X)):=\{\text{bounded linear operators $L^2(X)\to L^2(X)$}\}
$$
on $X$ is defined via functional calculus. Moreover, one has the mapping property
$$
e^{-t H}:L^2(X)\longrightarrow C(X)\quad\text{ for all $t>0$},
$$
and there exists a uniquely determined continuous map
$$
(0,\infty)\times X\times X \ni (t,x,y)\longmapsto p(t,x,y)\in [0,\infty),
$$
the \emph{heat kernel of $H$}, with the following property: for all $f\in L^2(X)$, $t>0$, $x\in X$, one has 
	$$
	e^{-t H}f(x) =\int_X p(t,x,y)f(y) \m(dy). 
$$
In addition, the heat kernel has the following properties for all $s,t>0$, $x,y\in X$:
\begin{itemize}
	\item symmetry: $p(t,y,x)= p(t,x,y)$,
	\item Chapman-Kolmogorov: $p(t+s,x,y) =\int_X p(t,x,z)p(s,z,y)\m(dz),$
	\item conservativeness: $\int_X p(t,x,z)\m(dz)=1$.
\end{itemize}

The asserted regularity facts on the heat kernel follow from abstract results on Dirichlet spaces obtained by Sturm in \cite{sturm2,sturm3}, in combination with the validity of appropriate local Poincar\'{e} inequalities in finite dimensional $\mathrm{RCD}^*$ spaces. The latter have been established by Rajala \cite{rajala}. The asserted conservativeness follows from the abstract criterion given in \cite{sturm1} and the validity of the Bishop-Gromov volume estimate.

\begin{Lemma}\label{assa} There exists a unique diffusion $\PP$ on $X$ such that for all $x_0\in X$, $n\in\IN$, $0<t_1< \dots< t_n$, and all Borel sets $A_1,\dots,A_n\subset M$ one has
\begin{align}\nonumber
&\PP^{x_0}\{\omega: \omega(t_1)\in A_1,\dots,\omega(t_n)\in A_n\}\\\label{shgi}
& =\int\cdots \int 1_{A_1}(x_1)p(\delta_0 ,x_0,x_1) \cdots
 1_{A_n}(x_n) p(\delta_{n-1} ,x_{n-1},x_n) \m(dx_1)\cdots \m(dx_n),
\end{align}
where $\delta_j:=t_{j+1}-t_j$, $t_0:=0$. Moreover, for every $x\in X$, $\alpha<1/2$, the measure $\PP^x$ is concentrated on locally $\alpha$-H\"older continuous paths and is called \emph{Brownian motion measure} or \emph{Wiener measure} with initial point $x$.
\end{Lemma}

\begin{Remark}\label{asllp}
As $\IEE$ is a regular local Dirichlet form, it follows immediately from Fukushima's theory \cite{fuku} that we can pick a diffusion $\PP$ on $X$ which satisfies (\ref{shgi}) for $\IEE$-quasi every $x_0\in X$. However, from the author's point of view, this approach would make the formulation of Proposition \ref{kazu} b) and of our main result Theorem \ref{haupt1} below somewhat artificial. Such a 'quasi-sure' formulation could even be carried out for $\mathrm{RCD}(K,\infty)$ spaces using the Ma-R\"ockner correspondence \cite{roeck} between local quasi-regular Dirichlet forms and diffusions, as the main ingredient of our machinery (which is Corollary \ref{am} below) is an $L^{\infty}$-to-Lipschitz smoothing result for the heat semigroup by Ambrosio/Gigli/Savar\'{e} \cite{ags}, which is valid in the infinite dimensional case $N=\infty$, too (and also for the fractional metric measure spaces considered in \cite{alonso}). Instead we establish the correspondence of $\IEE$ with a \emph{pointwise} uniquely determined diffusion by hand, using classical methods from Markov processes in combination with a Li-Yau heat kernel upper estimate for $\mathrm{RCD}^*(K,N)$ spaces, 
$$
p(t,x,y)\leq C \m\big(B(x,\sqrt{t})\big)^{-1} e^{-\frac{\mathfrak{d}(x,y)^2}{Ct}}\quad\text{for all $0<t<1$, $x,y\in X$,}
$$
which has been obtained by Jiang/Li/Zhang \cite{jiang}, and which in general needs $N<\infty$. We would like to mention in this context that a variant of the Li-Yau heat kernel bound for $\mathrm{RCD}(K,\infty)$ spaces, namely an estimate of the form
$$
p(t,x,y)\leq C \frac{ e^{-\frac{\mathfrak{d}(x,y)^2}{Ct}}}{\sqrt{\m\big(B(x,\sqrt{t})\big)\m\big(B(y,\sqrt{t})\big)}}\quad\text{for all $0<t<1$, $x,y\in X$,}
$$ 
has been recently established by Tamanini in \cite{taman}. It seems, however, that this estimate does not suffice to carry out the arguments given in the proof below. The point is that in the infinite dimensional case, one cannot replace 
$$
1/\sqrt{\m\big(B(x,\sqrt{t})\big)\m\big(B(x,\sqrt{t})\big)}
$$
by
$$
\mathrm{const.} \times  \m\big(B(x,\sqrt{t})\big)^{-1},
$$
which in the finite dimensional case can be done using the Bishop-Gromov volume estimate.
\end{Remark}

\begin{proof}[Proof of Lemma \ref{assa}] The existence of a unique probability measure $\PP^{x_0}$ on $C([0,\infty),X)$ satisfying (\ref{shgi}) and giving full measure on $\alpha$-H\"older continuous paths follows from Kolmogorov's theorems on consistency and the existence of a (H\"older)-continuous modification, if we can show that for all $T>0$ there exists $C_T>0$, such that for all $0<t_1<t_2\leq T$ one has
\begin{align}\label{swayy}
\int_Xp(t_1,x_1,x_0)\int_X \mathfrak{d}(x_1,x_2)^2 p(t_2-t_1,x_1,x_2)  \m(dx_2)\m(d x_1)\leq C_T (t_2-t_1)^2.
\end{align}
In order to show (\ref{swayy}), we are going to prove that there exists $C_T>0$ such that for all $x\in X$, $0<t\leq T$ one has  
$$
\int_X \mathfrak{d}(x,y)^2 p(t,x,y)  \m(dy)\leq C_T t^2.
$$
The Li-Yau upper heat kernel estimate implies the existence of $C'_T>0$ (which only depends on $T$, $N$ and $K$) and $C>0$ (which only depends on $N$ and $K$) such that for all $x,y\in X$, $0<t< T$ one has
$$
p(t,x,y)\leq C'_T \m(x,\sqrt{t})^{-1} e^{-\frac{\mathfrak{d}(x,y)^2}{Ct}},\quad\text{with $\m(x,\sqrt{t}):= \m\big(B(x,\sqrt{t})\big)$.}
$$
Moreover, one has the following local volume doubling property: there exists a constant $C''_T>0$ (which only depends on $T$, $K$, $N$) such that for all $0<t<T$, $x\in X$, one has $\m(x,2t)\leq C''_T\m(x,t)$. As pointed out in \cite{kuwada}, local volume doubling follows from the Bishop-Gromov volume inequality, which also implies that $X$ is locally compact. By a standard argument (cf. Lemma 5.27 in \cite{saloff-buch}), local volume doubling implies the existence of constants $\nu_T\geq 1$, $C'''_T, C''''_T>0$ (which only depend on $C''_T$), such that for all $s,t>0$, $x\in X$ one has
\begin{align}\label{salo}
\frac{\m(x,t)}{\m(x,s)}\leq C'''_T \left(\frac{t}{s}\right)^{\nu_T} e^{C''''_T \frac{t}{T}}.
\end{align}
Now fix arbitrary $0<t\leq T$, $x\in X$. Then Li-Yau implies
\begin{align*}
\int_X \mathfrak{d}(x,y)^2 p(t,x,y)  \m(dy)&\leq C'_T\m(x,\sqrt{t})^{-1}\int_{B(x,t)}\mathfrak{d}(x,y)^2 e^{-\frac{\mathfrak{d}(x,y)^2}{Ct}}\m(dy)\\
&\quad+C'_T\m(x,\sqrt{t})^{-1}\int_{X\setminus B(x,t)}\mathfrak{d}(x,y)^2 e^{-\frac{\mathfrak{d}(x,y)^2}{Ct}}\m(dy).
\end{align*}
Using (\ref{salo}) we find
\begin{align*}
&\m(x,\sqrt{t})^{-1}\int_{B(x,t)} \mathfrak{d}(x,y)^2e^{-\frac{\mathfrak{d}(x,y)^2}{Ct}}\m(dy)\leq t^2\frac{\m(x,t)}{\m(x,\sqrt{t})}\leq C'''_Tt^2t^{\nu_T/2} e^{C''''_T \frac{t}{T}},
\end{align*}
which has the desired form.\\
For the second integral we have
\begin{align*}
&\m(x,\sqrt{t})^{-1}\int_{X\setminus B(x,t)}\mathfrak{d}(x,y)^2 e^{-\frac{\mathfrak{d}(x,y)^2}{Ct}}\m(dy)\\
&\leq \m(x,\sqrt{t})^{-1}\sum^{\infty}_{k=1}\int_{\{y\in X:kt\leq \mathfrak{d}(x,y)\leq (k+1)t\}} \mathfrak{d}(x,y)^2e^{-\frac{\mathfrak{d}(x,y)^2}{Ct}}\m(dy)\\
&\leq t^2 \sum^{\infty}_{k=1}\frac{\m(x,(k+1)t)}{\m(x,\sqrt{t})} (k+1)^2 e^{-\frac{tk^2}{C}}\leq C'''_Tt^2 t^{\frac{\nu_T}{2}} \sum^{\infty}_{k=1} (k+1)^{2+\nu_T} e^{-\frac{tk^2}{C}+\frac{C'''_Tt(k+1)}{T}}\\ 
&\leq C'''_Te^{\frac{C'''_Tt}{T}}t^2 t^{\frac{\nu_T}{2}} \sum^{\infty}_{k=1}  e^{-\frac{tk^2}{C}+ \left(2+\nu_T+\frac{C'''_Tt}{T}\right)k}\leq C'''_Te^{\frac{C'''_Tt}{T}}t^2 t^{\frac{\nu_T}{2}} \int_{\IR}  e^{-\frac{tr^2}{C}+ \left(2+\nu_T+\frac{C'''_Tt}{T}\right)r} dr\\
&= C'''_Te^{\frac{C'''_Tt}{T}}t^2 t^{\frac{\nu_T}{2}-\frac{1}{2}}\sqrt{\pi C}e^{ \frac{1}{4t} C \left(2+\nu_T+\frac{C'''_Tt}{T}\right)^2  },
\end{align*}
which in view of $\nu_T\geq 1$ again has the desired form. Above we have used (\ref{salo}) and the trivial inquality $(1+a)^{\nu}\leq e^{\nu a}$, $a\in \IR$, $\nu>0$.\\
Having established the existence of a unique family of probability measures $\PP=(\PP^z)_{z\in X}$ on $C([0,\infty),X)$ satisfying (\ref{shgi}) for all $x_0\in X$, note that by a montone class argument one has
\begin{align}\label{dwaaa}
\int f(\omega(t)) \PP^x( d \omega)= \int_X p(t,x,y) f(y) \m(dy)\quad\text{for all bounded $f:X\to \IR$, $t>0$, $x\in X$.}
\end{align}
Then
$$
\lim_{t\to 0+} \int f(\omega(t)) \PP^x( d \omega)= f(x)\quad\text{ for all bounded continuous $f:X\to \IR$}
$$
follows from Proposition 3.2 in \cite{ags}, showing that that $\PP^x$ is concentrated on paths starting from $x$, and the Chapman-Kolmogorov equation immediately imply that $\PP$ has the Markov property. The strong Markov property then follows from the continuity of $x\mapsto \int f(\omega(t)) \PP^x( d \omega)$ for all $t\geq 0$ and all bounded continuous $f:X\to\IR$, which is a consequence of (\ref{dwaaa}) and Proposition 3.2 in \cite{ags}. 
\end{proof}

Note that by construction we have the absolute continuity property
$$
\m(A)=0 \>\>\Rightarrow \>\>\int^t_0\PP P^x\{\omega:\omega(s)\in A\} ds=0\quad\text{for all Borel subsets $A\subset X$, $t\geq 0$.}
$$

For obvious reasons, any coupling of $\PP$ is called a \emph{coupling of Brownian motions on $X$}.

\begin{Proposition}\label{am} For all $x,y\in X$ there exists a coupling 
\begin{align}\label{abs3}
(\mathbb{X},\mathbb{Y}): (\Omega, \IFF,P)\longrightarrow C([0,\infty), X\times X)
\end{align}
of Brownian motions on $(X,\mathfrak{d},\m)$ from $x$ to $y$ with 
	\begin{align}\label{gut}
P   \{ \tau(\mathbb{X},\mathbb{Y})>t\}\leq \frac{1}{2}F_K(t)\mathfrak{d}(x,y)\quad\text{ for all $t>0$},
\end{align}
where 
$$
F_K(t):=\begin{cases} &\frac{1}{\sqrt{2t}},\quad\text{if $K=0$}\\
&\sqrt{ \frac{K}{e^{2Kt}-1}  } ,\quad\text{ if $K\ne 0$,}\end{cases}
$$
in particular, every such coupling is successful if $K\geq 0$, and one has the following global H\"older estimate: for all $\alpha\in (0,1]$, all real-valued $f\in L^{\infty}(X,\m)$, $t>0$,
\begin{align}<
\left|e^{-tH}f(x)-e^{-tH}f(y)\right|\leq  2^{1-\alpha}F_K(t)^{\alpha} \mathfrak{d}(x,y)^{\alpha}\left\|f\right\|_{L^{\infty}}.
\end{align}
\end{Proposition}

\begin{proof} In view of Proposition \ref{kazu} and (\ref{dwaaa}), the claim follows immediately from the following Lipschitz smoothing result (cf. Theorem 3.17 in combination with Theorem 4.17 in \cite{ags}): one has
$$
\left|e^{-tH}f(x)-e^{-tH}f(y)\right|\leq F_K(t) \mathfrak{d}(x,y)\left\|f\right\|_{L^{\infty}}.
$$
\end{proof}

The following class of potentials will be the perturbations of $H$ that will be considered in the sequel:

\begin{Definition}\label{kat} Given $\alpha\in [0,1]$, a Borel function $V:X\to \IC$ is said to be in the \emph{$\alpha$-Kato class $\mathcal{K}^{\alpha}(X)$ of $X$}, if  
$$
\lim_{t\to 0+}\sup_{x\in X}\int^t_0 s^{-\alpha/2}\int_Xp (s,x,y) |V(y)|\m(dy)  ds=0.
$$

\end{Definition}

Note that the Kato property only depends on the $\m$-equivalence class induced by $V$, and that 
\begin{align}\label{ui}
\int^t_0 s^{-\alpha/2}\int_Xp (s,x,y) |V(y)|\m(dy)  ds=\int^t_0 s^{-\alpha/2}\int |V(\omega(s))|\PP^x(d\omega)  ds.
\end{align}

Each $\mathcal{K}^{\alpha}(X)$ is a linear space and $\mathcal{K}(X):=\mathcal{K}^{0}(X)$ is the usual Kato class. One has
\begin{align}
&\mathcal{K}^{\alpha}(X)\subset \mathcal{K}^{\beta}(X),\quad\text{if $\alpha\geq \beta$},\\
&L^{\infty}(X)\subset  \mathcal{K}^{\alpha}(X).
\end{align}

The following lemma allows to test the Kato assumption in typical applications:

\begin{Lemma}\label{guum} a) Let $\alpha\in [0,1]$ and $q\in [1,\infty)$ with $q> N/(2-\alpha)$ one has 
$$
L^q_{1/\m}(X)+L^{\infty}(X)\subset \mathcal{K}^{\alpha}(X),
$$
where 
$$
L^q_{1/\m}(X):= \left\{W: \int_X\frac{|W(x)|^q}{\m(B(x,1))}\m(dx)<\infty\right\}.
$$
b) Assume $M,\widetilde{M}$ are geodesically complete smooth Riemannian manifolds with Ricci curvature $\geq K$ and let $\pi:\widetilde{M}\to M$ be a smooth surjective Riemannian submersion such that fibers $\pi^{-1}(y)\subset \widetilde{M}$ are minimal submanifolds for all $y\in M$. Then for all Borel functions $\phi:C([0,\infty),M)\to [0,\infty]$, $x\in \widetilde{M}$, using an obvious notation, one has 
$$
\int \phi(\pi(\omega)) \widetilde{\PP}^x(d\omega)\leq \int \phi(\omega) \PP^{\pi(x)}(d\omega).
$$
In particular, for all $\alpha\in [0,1]$ one has $\pi^*\mathcal{K}^{\alpha}(M)\subset \mathcal{K}^{\alpha}(\widetilde{M})$.
\end{Lemma}

\begin{proof} a) It suffices to show that $V\in \mathcal{K}^{\alpha}(X)$ for all $V\in L^q_{1/\m}(X)$. This can be seen as follows: with $q'$ the Hölder conjugate of $q$, for all $0<s<1$, $x\in M$ one has,
\begin{align*}
&\int_X p(s,x,y)|V(y)| \m(dy)=\int_X p(s,x,y)^{\f{1}{q'}} p(s,x,y)^{1-\f{1}{q'}}|V(y)| \m(dy)\\
&\leq \Big(\int_X p(s,x,y)\m(dy)\Big)^{\f{1}{q'}}\left(\int_X |V(y)|^{q}p(s,x,y)\m(dy)\right)^{\f{1}{q}}\\
&\leq C(K,N,q)\big(1+s^{-\frac{N}{2q}}\big) \left(\int_X |V(y)|^{q}\m(B(x,1))^{-1}\m(dy)\right)^{\f{1}{q}},
\end{align*}
where the second estimate follows from
$$
\int_X p(s,x,y)\mu(dy)\leq 1,
$$
in combination with the Li-Yau heat kernel bound and with Bishop-Gromov's volume estimate (cf. Example IV.18 in \cite{Batu} for a detailed argument). \\
b) The asserted estimate has been shown in \cite{BatuHeat}, and the asserted inclusion is a trivial consequence of the estimate.
\end{proof}

Note that $L^q_{1/\m}(X)\subset L^q(X)$ by the Bishop-Gromov inequality, while the other inclusion typically requires $\inf_{x\in X}\m(B(x,1))>0$. It is certainly an interesting problem to check how possible further properties of the Kato classes (for example, uniformly local $L^q$-conditions for the Kato classes, or Green's function characterizations of these spaces) depend on the geometry of $X$. Here, the methods developed by Kuwae and Takahashi in \cite{kt} should be applicable in principle (see also \cite{aizen,Batu}).\vspace{2mm}

The Kato condition is linked to operator theory via the following fact: for all real-valued $V\in \mathcal{K}(X)$, and all $\epsilon>0$ there exists $C_{\epsilon}<\infty$, which depends on $\epsilon$ (and $V$), such that \cite{Batu, peter}
$$
\int_XV(x)  |f(x)|^2\m(dx)\leq \epsilon\IEE(f)+C_{\epsilon}\left\|f\right\|^2_{L^2}\quad\text{for all $f\in \dom(\IEE)$.}
$$
In the language of perturbation theory this estimate means that the symmetric sesquilinear form in $L^2(X)$ induced by $V$ is infinitesimally $\IEE$-bounded, and so the KLMN-theorem \cite{teschl} implies that the symmetric sesquilinear form 
$$
\dom(\IEE)\times \dom(\IEE)\ni (f_1,f_2)\longmapsto \IEE_V(f_1,f_2):=\IEE(f_1,f_2)+ \int_XV(x)  f_1(x)f_2(x)\m(dx)\in \IC
$$
is semibounded from below and closed (it is densely defined, as $\IEE$ is so). Thus $\IEE_V$ canonically induces a self-adjoint semibounded from below operator $H_V$ in $L^2(X)$. \vspace{1mm}

Let 
$$
(e^{-tH_V})_{t\geq 0}\subset \ILL(L^2(X))
$$
denote the \emph{Schr\"odinger semigroup} induced by $V$. For all $V\in \mathcal{K}(X)$ one has the Feynman-Kac formula, which states that for all $\Psi\in L^2(X)$, $t> 0$ and $\m$-a.e. $x\in X$ one has
$$
e^{-tH_V}\Psi(x)= \int e^{-\int^t_0V(\omega(s))ds}\Psi(\omega(t))\PP^x(d\omega).
$$
The latter formula can be proved precisely as in the Riemannian case \cite{Batu}. If $\Psi\in L^{q}(X)$, $q\in [1,\infty]$, then we take the RHS of the Feynman-Kac formula as a pointwise well-defined representative of $e^{-tH_V}\Psi(x)$, which as a function of $x$ defines an element on $L^{q}(X)$.


Here comes our main result:

\begin{Theorem}\label{haupt1} Let $\alpha\in (0,1]$, $V\in \mathcal{K}^{\alpha}(X)$, $t>0$. \\
a) For all $x,y\in X$, $\Phi\in L^{\infty}(X)$ one has
$$
\left|e^{-tH_{V}}\Phi(x)-e^{-tH_{V}}\Phi(y)\right|\leq \Big(2^{1-\alpha}F_K(t)^{\alpha}+ A(V,K,\alpha,t)\Big) \left\|\Phi\right\|_{L^{\infty}}\mathfrak{d}(x,y)^{\alpha},
$$
where 
$$
A(V,K,\alpha,t):=2^{2-\alpha}\sup_{x\in X}\int e^{-\int^{t}_0V(\omega(s))ds}\PP^x(d\omega) \cdot \int^{t/2}_0F_K(s)^{\alpha}\int_X p(s,x,y)|V(y)|\m(dy) ds<\infty.
$$
b) For all $x,y\in X$, $\lambda\in\IR$ and all $\Psi\in \dom(H_V)\cap L^{\infty}(X)$ with $H_V\Psi=\lambda\Psi$ one has 
$$
\left|\Psi(x)-\Psi(y)\right|\leq e^{t\lambda}\Big(2^{1-\alpha}F_K(t)^{\alpha}+ A(V,K,\alpha,t\Big) \left\|\Psi\right\|_{L^{\infty}}\mathfrak{d}(x,y)^{\alpha}.
$$
\end{Theorem}

\begin{proof} a) Given $r>0$ set
\begin{align*}
&C_{\exp}(V,r):=\sup_{x\in X}\int e^{-\int^r_0V(\omega(s))ds}\PP^x(d\omega),\\
&C(V,K,\alpha,r):=\sup_{x\in X}\int^r_0F_K(s)^{\alpha}\int_X p(s,x,y)|V(y)|\m(dy) ds,
\end{align*}
so that
$$
A(V,K,\alpha,r)=2^{2-\alpha}C(V,K,\alpha,r/2)C_{\exp}(V,r).
$$
We have 
\begin{align}\label{dede}
C_{\exp}(V,r)\leq 2 e^{C_Vr}<\infty
\end{align}
 by the so called Khashminskii's lemma, which can be proved precisely as in the Riemannian case \cite{aizen,Batu}. In order to show $C(V,K,\alpha,r)<\infty$, as $F_K$ is bounded on each compact subset of $[0,\infty)$ and as $F_K^{\alpha}$ behaves like $s^{-\alpha/2}$ near $s=0$, it suffices to show
\begin{align*}
\sup_{x\in X}\int^t_0s^{-\alpha/2}\int_X p(s,x,y)|V(y)|\m(dy) ds<\infty.
\end{align*}
We can follow the argument from \cite{kw0}: pick a $t'>0$ with 
\begin{align*}
\sup_{x\in X}\int^{t'}_0s^{-\alpha/2}\int_X p(s,x,y)|V(y)|\m(dy) ds<\infty
\end{align*}
and an $l\in\IN$ with $t<lt'$. Then we can estimate 
\begin{align*}
&\sup_{x\in X}\int^{t}_0s^{-\alpha/2}\int_X p(s,x,y)|V(y)|\m(dy) \ ds\nn\\
&\leq \sup_{x\in X} \int_X\int^{l t'}_0 s^{-\alpha/2}p(s,x,y)|V(y)|  d s \ \m(dy) \nn\\
&=\sup_{x\in X} \int_X\sum^l_{k=1}\int^{ t'}_0 ((k-1)t'+s)^{-\alpha/2}p((k-1)t'+s,x,y)|V(y)|  d s \ \m(dy)\nn\\
&\leq\sum^l_{k=1}\sup_{x\in X} \int_X\int^{ t'}_0 s^{-\alpha/2}p((k-1)t'+s,x,y)|V(y)|  d s \ \m(dy)\nn\\
&=\sum^l_{k=1}\sup_{x\in X} \int_X p((k-1)t',x,z) \int^{ t'}_0s^{-\alpha/2} \int_X p(s,z,y)|V(y)|  \m(dy) \ d s  \  \m(dz)\nn\\
&\leq \sup_{z\in X}\int^{ t'}_0s^{-\alpha/2}\int_X p(s,z,y)|V(y)| \m(dy) \ d s  \ \left(\sum^l_{k=1} \sup_{x\in X}\int_X p((k-1)t',x,z)\m(dz)\right)\\ 
&\leq l \sup_{z\in X}\int^{ t'}_0s^{-\alpha/2}\int_X p(s,z,y)|V(y)| \m(dy) \ d s <\infty,
\end{align*}
where we have used the Chapman-Kolomogorov identity and $\int_X p((k-1)t',x,z) \m(dz)\leq 1$. In order to prove the asserted H\"older estimate, we denote the seminorm on the space $C^{0,\alpha}(X)$ of $\alpha$-H\"older continuous functions $f:X\to\IC$ by
$$
\left\|f\right\|_{0,\alpha}= \sup_{x\ne y} |f(x)-f(y)|\mathfrak{d}(x,y)^{-\alpha}.
$$
Let $\Phi\in L^{\infty}(X)$ and define 
$$
V_{n}:=\max(V,n),\quad V_{n,m}:=\min(m,V_n)
$$
for $m,n\in\IN$. Duhamel's formula states that
\begin{align*}
e^{-tH_{V_{n,m}}}\Phi=e^{-tH}\Phi+\int^{t}_0 e^{-\frac{s}{2}H}e^{-\frac{s}{2}H}V_{n,m}e^{-(t-s)H_{V_{n,m}}}\Phi ds,
\end{align*}
a fact which can proved using the Markov property of $\PP$. Thus we have
\begin{align}\label{v1}
&\left\|e^{-tH_{V_{n,m}}}\Phi\right\|_{C^{0,\alpha}}\leq \left\|e^{-tH}\Phi\right\|_{C^{0,\alpha}}\\
&\quad\quad\quad\quad+\int^{t}_0 \left\|e^{-\frac{s}{2}H}\right\|_{L^{\infty}\to C^{0,\alpha}}\left\|e^{-\frac{s}{2}H}V_{n,m}\right\|_{L^{\infty}\to L^{\infty}}\left\|e^{-(t-s)H_{V_{n,m}}}\Phi \right\|_{L^{\infty}}ds.
\end{align}
Proposition \ref{am} implies  
\begin{align}\label{v2}
\left\|e^{-tH}\Phi\right\|_{C^{0,\alpha}}\leq 2^{1-\alpha}F_K(t)^{\alpha} \left\|\Phi\right\|_{L^{\infty}},
\end{align}
and
\begin{align}\label{v3}
\left\|e^{-\frac{s}{2}H}\right\|_{L^{\infty}\to C^{0,\alpha}}\leq 2^{1-\alpha}F_K(s/2)^{\alpha}. 
\end{align}
By the Feynman-Kac formula we have
\begin{align}\label{v4}
\left\|e^{-(t-s)H_{V_{n,m}}}\Phi \right\|_{L^{\infty}}\leq C_{\exp}(V,t)\left\|\Phi \right\|_{L^{\infty}}.
\end{align}
Finally, given $f\in L^{\infty}(X)$ one has
\begin{align*}
&\int^t_0\left\|e^{-\frac{s}{2}H}\right\|_{L^{\infty}\to C^{0,\alpha}}\left\|e^{-\frac{s}{2}H}V_{n,m}f\right\|_{L^{\infty}}ds\\
&\leq  2^{1-\alpha} \int^t_0     F_K(s/2)^{\alpha}\sup_x\int_Xp(s/2,x,y) |V_{n,m}(y)|\m(dy)   ds  \left\|f\right\|_{L^{\infty}}\\
&\leq 2^{2-\alpha} \sup_x \int^{t/2}_0   F_K(s)^{\alpha}\int_Xp(s,x,y) |V_{n,m}(y)|\m(dy)   ds  \left\|f\right\|_{L^{\infty}}\\
&\leq 2^{2-\alpha} C(V,K,\alpha,t/2)  \left\|f\right\|_{L^{\infty}},
\end{align*}
where the second inequality can be seen as follows: pick a sequence $x_l$ in $X$ such that 
\begin{align*}
&F_K(s/2)^{\alpha}\int_Xp(s/2,x_l,y) |V_{n,m}(y)|\m(dy)\leq F_K(s/2)^{\alpha}\int_Xp(s/2,x_{l+1},y) |V_{n,m}(y)|\m(dy),\\
&\sup_x F_K(s/2)^{\alpha}\int_Xp(s/2,x,y) |V_{n,m}(y)|\m(dy)=\lim_l  F_K(s/2)^{\alpha}\int_Xp(s/2,x_l,y) |V_{n,m}(y)|\m(dy),
\end{align*}
and so
\begin{align*}
&\int^t_0 F_K(s/2)^{\alpha}\sup_x\int_Xp(s/2,x,y) |V_{n,m}(y)|\m(dy) ds \\
&= \int^t_0   \lim_l  F_K(s/2)^{\alpha}\int_Xp(s/2,x,y) |V_{n,m}(y)|\m(dy)   ds\\
&= \lim_l\int^t_0     F_K(s/2)^{\alpha}\int_Xp(s/2,x,y) |V_{n,m}(y)|\m(dy)   ds\\
&\leq \sup_x\int^t_0     F_K(s/2)^{\alpha}\int_Xp(s/2,x,y) |V_{n,m}(y)|\m(dy)   ds\\
&=2\sup_x\int^{t/2}_0     F_K(s)^{\alpha}\int_Xp(s,x,y) |V_{n,m}(y)|\m(dy)   ds,
\end{align*}
by monotone convergence. Thus we have
\begin{align}\label{v5}
\int^{t}_0 \left\|e^{-\frac{s}{2}H}\right\|_{L^{\infty}\to C^{0,\alpha}}\left\|e^{-\frac{s}{2}H}V_{n,m}\right\|_{L^{\infty}\to L^{\infty}}ds\leq 2^{2-\alpha} C(V,K,\alpha,t/2).  
\end{align}
Putting together (\ref{v1})-(\ref{v5}) we arrive at the following inequality: for all $x,y\in X$, $\Phi\in L^{\infty}(X)$ one has
\begin{align*}
&\left|e^{-tH_{V_{n,m}}}\Phi(x)-e^{-tH_{V_{n,m}}}\Phi(y)\right|\\
&\leq \Big(2^{1-\alpha}F_K(t)^{\alpha}+2^{2-\alpha} C(V,K,\alpha,t/2)C_{\exp}(V,t)\Big) \left\|\Phi\right\|_{L^{\infty}}\mathfrak{d}(x,y)^{\alpha}.
\end{align*}
It remains to prove 
$$
\lim_n\lim_m\left|e^{-tH_{V_{n,m}}}\Phi(x)-e^{-tH_{V_{n,m}}}\Phi(y)\right|= \left|e^{-tH_{V}}\Phi(x)-e^{-tH_{V}}\Phi(y)\right|.
$$
Here, by linearity, we can assume $\Phi\geq 0$. Then
$$
\lim_m\left|e^{-tH_{V_{n,m}}}\Phi(x)-e^{-tH_{V_{n,m}}}\Phi(y)\right|= \left|e^{-tH_{V_n}}\Phi(x)-e^{-tH_{V_n}}\Phi(y)\right|
$$
follows from the Feynman-Kac formula and monotone convergence, and 
$$
\lim_n\left|e^{-tH_{V_{n}}}\Phi(x)-e^{-tH_{V_{n}}}\Phi(y)\right|= \left|e^{-tH_{V}}\Phi(x)-e^{-tH_{V}}\Phi(y)\right|
$$
follows from the Feynman-Kac formula and dominated convergence (using Khashminskii's lemma), completing the proof of part a).\\ 
b) This follows from part a) using $e^{-tH_V}\Psi=e^{-t\lambda}\Psi$ by the spectral calculus.
\end{proof}  

Combining this result with Lemma \ref{guum} we immediately get:

\begin{Corollary}\label{sub} Let $M, \widetilde{M}$ be smooth geodesically complete Riemannian manifolds with Ricci curvature $\geq K$ and let $\alpha\in (0,1]$. Let $\pi_j, \pi_{ij}:\widetilde{M}\to M$ be a finite collection of smooth surjective Riemannian submersions such that the fibers $\pi^{-1}_j(y), \pi^{-1}_{ij}(y)\subset \widetilde{M}$ are minimal submanifolds for all $y\in M$, and let 
\begin{align*}
&\text{$V_j\in L^{q_j}_{1/\m}(M)+L^{\infty}(M)$ for some $q_j> \dim(M)/(2-\alpha)$,}\\
&\text{$V_{ij}\in L^{q_{ij}}_{1/\m}(M)+L^{\infty}(M)$ for some $q_{ij}> \dim(M)/(2-\alpha)$.}
\end{align*}
Then one has
$$
V:=\sum_{j} V_j\circ \pi_j+\sum_{i,j} V_{ij}\circ \pi_{ij}\in \mathcal{K}^{\alpha}(\widetilde{M}),
$$
and for all $x,y\in \widetilde{M}$, $\Phi\in L^{\infty}(\widetilde{M})$ it holds that
$$
\left|e^{-tH_{V}}\Phi(x)-e^{-tH_{V}}\Phi(y)\right|\leq \Big(2^{1-\alpha}F_K(t)^{\alpha}+ B(V,K,\alpha,t)\Big) \left\|\Phi\right\|_{L^{\infty}}\mathfrak{d}(x,y)^{\alpha},
$$
where
\begin{align*}
&B(V,K,\alpha,t):=2^{2-\alpha}\sup_{x\in M}\int e^{-\sum_j\int^t_0V_j(\omega(s))ds-\sum_{ij}\int^t_0V_{ij}(\omega(s))ds}\PP^x (d\omega)\\
&\quad\quad\quad\quad\quad\quad\quad\times \left( \sum_j\int^{t/2}_0F_K(s)^{\alpha}\int_M p(s,x,y)|V_j(y)|\m(dy) ds\right.\\
&\quad\quad\quad\quad\quad\quad\quad\quad\quad+\left.\sum_{i,j}\int^{t/2}_0F_K(s)^{\alpha}\int_M p(s,x,y)|V_{ij}(y)|\m(dy) ds\right)<\infty.
\end{align*}
\end{Corollary}

Note that in the above situation $H_V$ is the unique self-adjoint realization of $-\Delta+V$ in $L^2(\widetilde{M})$ (cf. \cite{Batu} for the asserted essential self-adjointness).

\section{Application to molecular Schr\"odinger operators}

Assume $M=\IR^{3}$, $\widetilde{M}=\IR^{3m}$. Pick $l\in\IN$, $R\in \IR^{3l}$, $Z\in [0,\infty)^l$ and consider the potential 
$$
V_{R,Z}:\IR^{3m}\longrightarrow \IR,\quad V_{R,Z}(\mathbf{x}_1, \dots,\mathbf{x}_m):=-\sum_{j=1}^m \sum_{i=1}^l \frac{Z_i}{|\mathbf{x}_j-\mathbf{R}_i|}+ \sum_{1\leq i<j\leq m} \frac{1}{|\mathbf{x}_i-\mathbf{x}_j|}
$$
of a molecule having $l$ nuclei and $m$ electrons, where the $j$-th nucleus carrying $\sim Z_j$ protons is considered to be fixed in $\mathbf{R}_j$ (infinite mass limit). The elementary charge has been set equal to $1$. Then with
\begin{align*}
&\pi_j:\IR^{3m}\longrightarrow \IR^3, \quad (\mathbf{x}_1, \dots,\mathbf{x}_m)\longmapsto \mathbf{x}_j, \\
&\pi_{ij}:\IR^{3m}\longrightarrow \IR^3, \quad (\mathbf{x}_1, \dots,\mathbf{x}_m)\longmapsto \mathbf{x}_i-\mathbf{x}_j, \\
&V_j(\mathbf{x}):=\sum_{i=1}^l Z_i/|\mathbf{x}-\mathbf{R}_i|,\quad \mathbf{x}\in\IR^3 \\
& V_{ij}(\mathbf{x}):=1/|\mathbf{x}|,,\quad \mathbf{x}\in\IR^3,
\end{align*}
one finds that the assumptions of the Corollary \ref{sub} are satisfied for all $\alpha\in (0,1)$, since
$$
V_j,V_{ij}\in L^{q}(\IR^{3})+L^{\infty}(\IR^{3})\quad\text{ for all $q\in [1,3)$.}
$$
In fact, one has $V_{R,Z}\in \mathcal{K}^{\alpha}(\IR^{3m})$, if and only if $\alpha<1$. In the sequel, $C_{a,b,\dots}>0$ denotes a constant which only depends on the parameters $a,b,\dots$, and which may change from line to line (while $C$ will denote any universal constant). Since\footnote{This mapping property relies on the fact that the Euclidean space is ultracontractive and on the Kato property of the potential.} \cite{Simon} 
\begin{align*}
e^{-\frac{t}{2}H_{V_{R, Z}}}: L^r(\IR^{3m})\longrightarrow L^{\infty}(\IR^{3m}),
\end{align*}
with
$$
\left\|e^{-\frac{t}{2}H_{V_{R, Z}}}\right\|_{ L^r(\IR^{3m})\to L^{\infty}(\IR^{3m})}\leq C t^{-\frac{3m}{2r}}e^{C_{R,Z}t},
$$
we obtain from 
$$
e^{-tH_{V_{R, Z}}}=e^{-\frac{t}{2}H_{V_{R, Z}}}e^{-\frac{t}{2}H_{V_{R, Z}}}
$$
and Corollary \ref{sub} the following smoothing property,
\begin{align*}
e^{-tH_{V_{R, Z}}}: L^r(\IR^{3m})\longrightarrow C^{0,\alpha}(\IR^{3m})\quad\text{ for all $t>0$, $r\in [1,\infty]$, $\alpha\in (0,1)$,}
\end{align*}
with
\begin{align*}
&\left\|e^{-tH_{V_{R, Z}}}\right\|_{ L^r(\IR^{3m})\to C^{0,\alpha}(\IR^{3m})}\leq \left\|e^{-\frac{t}{2}H_{V_{R, Z}}}\right\|_{ L^{\infty}(\IR^{3m})\to C^{0,\alpha}(\IR^{3m})}\left\|e^{-\frac{t}{2}H_{V_{R, Z}}}\right\|_{ L^r(\IR^{3m})\to  L^{\infty}(\IR^{3m})}\\
&\leq  C\Big(2^{1-\alpha}t^{-\alpha/2}+ B(V_{R,Z},K,\alpha,t/2)|_{K=0}\Big) t^{-\frac{3m}{2r}}e^{C_{R,Z}t}.
\end{align*}

Using 
$$
4\pi\int^{\infty}_{0} p(r,\mathbf{y},\mathbf{R})  dr = |\mathbf{y}-\mathbf{R}|^{-1},
$$
Fubini, Chapman-Kolmogorov and
$$
p(u,\mathbf{a},\mathbf{b})\leq C u^{-\frac{3}{2}},
$$
we get
\begin{align*}
& \int_{\IR^{3}} p(s,\mathbf{x},\mathbf{y})|\mathbf{y}-\mathbf{R}|^{-1} d\mathbf{y} =4\pi \int^{\infty}_{0} \int_{\IR^{3}}  p(s,\mathbf{x},\mathbf{y})  p(r,\mathbf{y},\mathbf{R})   d\mathbf{y} dr \\
&= 4\pi \int^{\infty}_{0} p(s+r,\mathbf{x},\mathbf{R})  dr\\ 
&\leq C \int^{\infty}_{0} (s+r)^{-3/2} dr\leq C s^{-1/2}, 
\end{align*}
and so using $F_K|_{K=0}(s)\leq s^{-1/2}$, 
\begin{align*}
&\sum_{j=1}^m\int^{t/4}_0F_K|_{K=0}(s)^{\alpha}\int_{\IR^3} p(s,\mathbf{x},\mathbf{y})|V_{j}(\mathbf{y})|d\mathbf{y} ds\\
&\leq C_{m,Z}  \int^{t/4}_0 s^{-\alpha/2-1/2} ds  \\
&=C_{m,Z} \frac{(t/4)^{-\alpha/2+1/2}}{1/2-\alpha/2},  
\end{align*}
and likewise
\begin{align*}
&\sum_{1\leq i<j\leq m}\int^{t/4}_0F_K|_{K=0}(s)^{\alpha}\int_{\IR^3} p(s,\mathbf{x},\mathbf{y})|V_{ij}(\mathbf{y})|d\mathbf{y} ds\\
&\leq C_m\cdot \frac{(t/4)^{-\alpha/2+1/2}}{1/2-\alpha/2},  
\end{align*}
and so, using (\ref{dede}),  
\begin{align*}
B(V_{R,Z},K,\alpha,t/2)|_{K=0}\leq C_{m,Z}  \frac{(t/4)^{-\alpha/2+1/2}}{1/2-\alpha/2}e^{C_{R,Z}t},
\end{align*}
we arrive at the estimate
\begin{align*}
\left\|e^{-tH_{V_{R, Z}}}\right\|_{ L^r(\IR^{3m})\to C^{0,\alpha}(\IR^{3m})}\leq  C_{m,Z}t^{-\frac{3m}{2r}}e^{C_{R,Z}t}\Big(2^{1-\alpha}t^{-\alpha/2}+   \frac{(t/4)^{-\alpha/2+1/2}}{1/2-\alpha/2}e^{C_{R,Z}t}\Big).
\end{align*}
 
As we have remarked in the introduction, this estimate shows the geometry behind the $\alpha$-Hölder continuity of molecular eigenfunctions, while it leaves open the question behind the geometry of the Lipschitz continuity of these eigenfunctions.


\end{document}